\documentclass{article}[12pt]

\usepackage[]{amsmath,amssymb,amsfonts,latexsym,amsthm,enumerate,cite,fullpage}

\newtheorem{thm}{Theorem}[section]
\newtheorem*{thm*}{Theorem}
\newtheorem{lemma}[thm]{Lemma}
\newtheorem*{lemma*}{Lemma}

\newtheorem*{prop*}{Proposition}
\newtheorem{cor}[thm]{Corollary}

\newtheorem{claim}[thm]{Claim}
\newtheorem{conj}[thm]{Conjecture}

\newtheorem{prob}[thm]{Problem}

\def\E{\mathbb{E}}
\def\F{\mathbb{F}}

\def\eps{\varepsilon}
\def\rank{\mathrm{rank}}
\def\poly{\mathrm{poly}}
\def\CC{\mathrm{CC}}
\def\CCdet{\CC^{\mathrm{det}}}
\def\bias{\mathrm{bias}}
\def\disc{\mathrm{disc}}
\newcommand{\ip}[1]{\left<#1\right>}

\newcommand\ignore[1]{}

\newcommand{\restate}[2]{\medskip \noindent {\bf #1 (restated)}{\sl #2}}

\begin{document}

\title{Recent advances on the log-rank conjecture in communication
  complexity}

\author{ Shachar Lovett\thanks{CSE department, UC San-Diego. e-mail:
    \texttt{slovett@cse.ucsd.edu}.}  }

\maketitle

\begin{abstract}
  The log-rank conjecture is one of the fundamental open problems in
  communication complexity. It speculates that the deterministic
  communication complexity of any two-party function is equal to the
  log of the rank of its associated matrix, up to polynomial
  factors. Despite much research, we still know very little about this
  conjecture. Recently, there has been renewed interest in this
  conjecture and its relations to other fundamental problems in
  complexity theory. This survey describes some of the recent
  progress, and hints at potential directions for future research.
\end{abstract}

\section{Introduction}
Communication complexity studies the amount of communication needed in
order to evaluate a function, whose output depends on information
distributed amongst two or more parties.  Since its first introduction
by Yao~\cite{yao1979some}, communication complexity was extensively
studied, to a large extent because of its applications in diverse
fields, such as circuit complexity, VLSI design, proof complexity,
streaming algorithms, data structures and more.  Still, there are many
fundamental problems about the communication complexity of functions
which are wide open.  We refer the reader to the book of Kushilevitz
and Nisan~\cite{KN97_Com} for more details on communication complexity
and its applications, and to the book of Lee and
Shraibman~\cite{lee2009lower} for an exposition of more recent lower
bound techniques in communication complexity.

In this survey, we focus on the communication complexity between two
parties.  Let $f:X \times Y \to \{0,1\}$ be a boolean function, where
one party holds an inputs $x \in X$, the other party holds an input $y
\in Y$, and their goal is to evaluate $f(x,y)$ while minimizing their
communication. For most of this survey, we will focus on deterministic
protocols, which is the simplest communication model. The
\emph{deterministic communication complexity} of $f$ is the minimal
number of bits communicated by an optimal deterministic protocol
computing $f$, and is denoted by $\CCdet(f)$.

There is a simple lower bound on the deterministic communication
complexity of functions, first observed by Mehlhorn and
Schmidt~\cite{mehlhorn1982vegas}, based on the rank of their
associated matrix.  Let $M_f$ be the $X \times Y$ matrix with $M_{x,y}
= f(x,y)$. A deterministic protocol computing $f$ in which the players
send $c$ bits of communication, corresponds to a partition of the
matrix $M_f$ to $2^c$ rectangles (a rectangle is a set $A \times B$
with $A \subset X, B \subset Y$) such that the value of $M_f$ is
constant on each rectangle. Such rectangles are called
\emph{monochromatic}. As the rank (as a real matrix) of a
monochromatic rectangle is at most one, we get that $\rank(M_f) \le
2^c$. Equivalently, if we shorthand $\rank(f)=\rank(M_f)$ then
$$
\CCdet(f) \ge \log\rank(f).
$$

The \textit{log-rank conjecture} proposed by Lov\'{a}sz and
Saks~\cite{LS88_Lat} speculates that this simple bound is tight for
all boolean functions, up to polynomial factors.

\begin{conj}[The log-rank conjecture~\cite{LS88_Lat}]\label{conj:logrank}
There exists a universal constant $C>0$ such that for any boolean function $f$,
$$
\CCdet(f) \le C (\log \rank(f))^C.
$$
\end{conj}

Validity of the log-rank conjecture is one of the fundamental open problems in communication complexity. It is true in all known examples, but still
very little progress has been made towards resolving it. In the special case where $M_f$ is the adjacency matrix of a graph $G$, an essentially equivalent conjecture
given by van Nuffelen~\cite{van1976bound} and Fajtlowicz~\cite{fajtlowicz1988conjectures} replaces the communication complexity
by the (weaker notion) of log of the chromatic number of the graph; equivalently, that $\chi(G) \le \exp(\log^{O(1)} \rank(G))$.

A simple upper bound is
that $\CCdet(f) \le \rank(f)$, which is exponentially worse than what is conjectured by the log-rank conjecture. It follows from the simple observation that
if $\rank(f)=r$, then there could be at most $2^r$ distinct rows in $M_f$. Hence, one can assume without loss of generality that $|X| \le 2^r$, and consider a protocol in which
the first player simply sends its input $x$. In the special case of graphs, Kotlov and Lov{\'a}sz~\cite{kotlov1996rank} proved that if a graph has rank $r$,
then its chromatic number is at most $2^{r/2}$. This was later improved to $(4/3)^r$ by Kotlov~\cite{K97_Rank}.

In terms of lower bounds, a sequence of works ~\cite{alon1989counterexample,razborov1992gap,raz1995log,NW94_On_Ra} culminating in
an example due to Kushilevitz (unpublished, cf.~\cite{NW94_On_Ra}) shows that there exist functions for which $\CCdet(f) \ge (\log \rank(f))^{^{\log_36}}$. Hence,the constant $C$ in Conjecture~\ref{conj:logrank}, if it exists, must satisfy $C \ge \log_36 \approx 1.63$.

Recently, there was renewed interest in the log-rank conjecture and its relations to several other problems in complexity theory. Ben-Sasson, Ron-Zewi and the author~\cite{BLR12_An_Ad} studied the relation of the log-rank
conjecture to the approximate duality conjecture of~\cite{zewi2011affine}, and showed that if one assumes a number-theoretic conjecture (the polynomial Freiman-Ruzsa conjecture) then the trivial upper bound
can be reduced by a logarithmic factor.

\begin{thm}[\cite{BLR12_An_Ad}]\label{thm:BLR_cc}
Assuming the polynomial Freiman-Ruzsa conjecture over $\F_2^n$, for any boolean function $f$,
$$
\CCdet(f)\le O(\rank(f) / \log \rank(f)).
$$
\end{thm}
Gavinsky and the author~\cite{GL13:lowrank_equiv} studied the relation between deterministic and randomized protocols for low rank matrices, and showed that in order to prove the log-rank conjecture, it suffices to
prove that any low rank matrix has an efficient randomized protocol. In fact, they show that even weaker notions of protocols are sufficient, like low information cost protocols or efficient zero-communication protocols. We will show here the following result.

\begin{thm}[\cite{GL13:lowrank_equiv}]\label{thm:GL}
If a boolean function $f$ has a randomized protocol of complexity $c$, then it also has a deterministic protocol of complexity $O(c \cdot \log^2(\rank(f)))$.
\end{thm}

Finally, the author~\cite{lovett2013communication} proved a new (unconditional) upper bound, based on discrepancy of low rank matrices, which improves the previous upper bound by nearly a quadratic factor.

\begin{thm}[\cite{lovett2013communication}]\label{thm:lovett}
For any boolean function $f$,
$$
\CCdet(f)\le O\left(\sqrt{\rank(f)} \cdot \log \rank(f)\right).
$$
\end{thm}

The goal of this survey is to explain these recent works, discuss their relations to other fundamental problems in complexity theory, and speculate on what directions seem the most likely to yield further advances for the log-rank conjecture. This is by no means a comprehensive survey. In particular, a related line of research which will not be discussed here is the study of the log-rank conjecture restricted to special families of functions. For example, the case of XOR functions (functions of the form $f(x,y)= F(x \oplus y)$) and related problems has received considerable attention recently~\cite{zhang2009communication,zhang2010parity,lee2010composition,montanaro2009communication,leung2011tight,sun2012randomized,liu2013quantum,
zhang2013efficient,shpilka2013structure}.

\paragraph{Paper organization.} In Section~\ref{sec:nw} we present a result of Nisan and Wigderson which allows to reduce the problem of constructing deterministic protocols to the simpler problem of exhibiting a large monochromatic rectangle. As this result is used repeatedly, we include its proof for completeness. In Section~\ref{sec:approx_dual} we discuss the approximate duality conjecture in additive combinatorics, its relations to the log-rank conjecture and to constructions of two-source extractors. In Section~\ref{sec:rand_det} we show that low-rank functions with efficient randomized protocols also have efficient deterministic protocols. In Section~\ref{sec:disc} we apply bounds on the discrepancy of low-rank functions to deduce better upper bounds on deterministic protocols. In Section~\ref{sec:further} we discuss several directions for further research, including relations to the problem of matrix rigidity.

\section{From monochromatic rectangles to protocols}
\label{sec:nw}
The log-rank conjecture speculates that if $M_f$ has a low rank, then
it can be partitioned into a small number of monochromatic
rectangles. In particular, it must have a large monochromatic
rectangle.  A beautiful reduction of Nisan and
Wigderson~\cite{NW94_On_Ra} shows that if one can prove that any low
rank boolean matrix has a large monochromatic rectangle, then it can
be bootstrapped to design a protocol with nearly the same
efficiency. As this reduction would be useful for us, we review it
below. We recall that a monochromatic rectangle is a subset $R=A
\times B \subset X \times Y$ such that $f(x,y)$ is constant for all
$(x,y) \in R$.

\begin{thm}[\cite{NW94_On_Ra}]\label{thm:NW}
Assume that for any function $f:X \times Y \to \{0,1\}$ with $\rank(f)=r$, there exists a monochromatic rectangle of size $|R| \ge 2^{-c(r)} |X \times Y|$. Then, any boolean function of rank $r$ is computable by a deterministic protocol of complexity $O(\log^2{r}+\sum_{i=0}^{\log{r}} c(r/2^i))$.
\end{thm}

Before giving the proof, we note that if $c(r) = \poly \log(r)$ then Theorem~\ref{thm:NW} implies a protocol with deterministic communication complexity $\poly \log(r)$, hence proving the log-rank conjecture. On the other end of the spectrum, if $c(r) = r^{\alpha}$ for some $\alpha<1$ then Theorem~\ref{thm:NW} implies a protocol with deterministic communication complexity $O(r^{\alpha})$.

\begin{proof}
Let $f$ be a function with $\rank(M_f)=r$, and let $R$ be the assumed monochromatic rectangle of size $2^{-c(r)} \cdot |X \times Y|$. Consider the partition of the matrix $M_f$ as
$$
M_f=\left(
\begin{array}{cc}
R&S\\
P&Q
\end{array}
\right)
$$
As $R$ is monochromatic, $\rank(R) \le 1$. Hence, $\rank(S)+\rank(P) \le r+1$. Assume, without loss of generality, that $\rank(S) \le r/2+1$ (otherwise, exchange the roles of the rows player and columns player). The row player sends one bit, indicating whether the input $x$ is in the top part or in the bottom part of the matrix. If it is in the top part then the rank decreases to $\rank(R \; S) \le \rank(R)+\rank(S) \le r/2+2$. If it is in the bottom part, the rank might not decrease, but the size of the matrix reduces to at most $(1-2^{-c(r)})|X \times Y|$. Iterating this process defines a protocol tree. We next bound the number of leaves of the protocol. By standard techniques, any protocol tree can be balanced so that the communication complexity is logarithmic in the number of leaves (cf.~\cite[Chapter 2, Lemma 2.8]{KN97_Com}).

Consider the protocol which stops once the rank drops to approximately $r/2$. The protocol tree in this case has at most $O(2^{c(r)} \cdot \log(|XY|))$ leaves, and hence can be simulated by a protocol sending only $O(c(r) + \log\log(|XY|))$ bits. Note that since we can assume $f$ has no repeated rows or columns, $|XY| \le 2^{2r}$ and hence $\log\log(|XY|) \le \log(r)+1$. Next, consider the phase where the protocol continues until the rank drops to $r/4$. Again, this protocol can be simulated by $O(c(r/2)+\log(r))$ bits of communication. Summing over $r/2^i$ for $i=0,\ldots,\log(r)$ gives the bound.
\end{proof}

\section{Approximate duality and the log-rank conjecture}
\label{sec:approx_dual}

Nisan and Wigderson~\cite{NW94_On_Ra} proved another interesting fact: any low rank boolean matrix contains a large rectangle which is slightly biased. The bias of $f$ over a rectangle $R$ is defined as
$$
\bias(f|R) = \left| \E_{(x,y) \in R} \left[(-1)^{f(x,y)}\right] \right| = \left| \Pr_{(x,y) \in R}[f(x,y)=0] - \Pr_{(x,y) \in R}[f(x,y)=1] \right|.
$$
We also define $\bias(f)$ to be the bias of $f$ over the full space $X \times Y$. We will later see a generalization of this fact, called discrepancy, which is measured against the worst case distribution
of inputs.

\begin{thm}[\cite{NW94_On_Ra}]\label{thm:NW_disc}
Let $f:X \times Y \to \{0,1\}$ with $\rank(f)=r$. Then there exists a rectangle $R$ of size $|R| \ge |X \times Y| / O(r^{3/2})$ such that $\bias(f|R) \ge 1 / O(r^{3/2})$.
\end{thm}

Let us restrict $f$ to the rectangle $R$ so that we may assume for simplicity $\bias(f) \ge \eps = 1/O(r^{3/2})$. Thus, we may ask whether it is easier to study the structure of low rank matrices, if we further
assume that they are somewhat biased. Recall that Theorem~\ref{thm:NW} requires us to find a large monochromatic rectangle. This raises the following problem.

\begin{prob}
Let $f$ be a boolean function such that $\rank(f)=r$ and $\bias(f) \ge \eps$. What is the largest monochromatic rectangle that $M_f$ must contain?
\end{prob}

The previous discussion shows that this problem is essentially equivalent to the log-rank conjecture, as long as the bias is inverse polynomially related to the rank. The main idea of Ben-Sasson et al.~\cite{BLR12_An_Ad} is to consider a related problem, where instead of considering the matrices over the reals, we consider them over the binary finite field $\F_2$. In the following, we denote by $\rank_{\F_2}(M_f)$ the rank of a matrix over $\F_2$; note that the rank over $\F_2$ is always at most the rank over the reals, e.g. $\rank_{\F_2}(M_f) \le \rank(M_f)$.

\paragraph{Approximate duality.}
We now introduce a seemingly unrelated problem. Let $A,B \subset \F_2^r$ be subsets. The \emph{approximate duality} measure of $A,B$ is
$$
\eps = \left| \E_{a \in A, b \in B}[(-1)^{\ip{a,b}}] \right| = \left| \Pr_{a \in A, b \in B}[\ip{a,b}=0] - \Pr_{a \in A, b \in B}[\ip{a,b}=1] \right|.
$$
We say the sets are $\eps$-approximate dual if their approximate duality measure is at least $\eps$. Note that $\eps=1$ corresponds to sets which are orthogonal (possibly after applying an affine shift to one of the sets). The \emph{approximate duality conjecture} of Ben-Sasson and Ron-Zewi~\cite{zewi2011affine} speculates that any large sets which are approximate dual, must contain large subsets which are dual.

\begin{conj}[Approximate duality conjecture~\cite{zewi2011affine}]\label{conj:approx_dual}
Let $A,B \subset \F_2^r$ be sets which are $\eps$-approximate dual. Then there exist subsets $A' \subset A, B' \subset B$ and a value $c \in \F_2$ such that
$$
\ip{a,b}=c \qquad \forall a \in A', b \in B',
$$
where
$$
\frac{|A|}{|A'|}, \frac{|B|}{|B'|} \le 2^{O\left(\sqrt{r \log(1/\eps)}\right)}.
$$
\end{conj}

The bound in Conjecture~\ref{conj:approx_dual}, if true, is the best possible, as the following example shows. Let $A=B$ be the set of all vectors in $\F_2^r$ of hamming weight $\sqrt{r}/10$. Then
the probability that a uniformly chosen $a \in A, b \in B$ intersect is at most $1/100$, and hence $A,B$ are $\eps$-approximate dual for $\eps \ge 0.98$. On the other hand, the largest subsets $A' \subset A, B' \subset B$
which are orthogonal come from choosing $A' = A \cap (\{0,1\}^{r/2} \times 0^{r/2})$ to be the set of vectors supported on the first half of the coordinates, and $B' = B \cap (0^{r/2} \times \{0,1\}^{r/2})$ to
be the vectors supported on the last half of the coordinates. One can then verify that $|A|/|A'|=|B|/|B'| = \exp(\Omega(\sqrt{r}))$. The bound for general $\eps>0$ can be similarly obtained,
by considering $A=B$ to be the vectors in $\F_2^r$ of hamming weight $O(\sqrt{r \log(1/\eps)})$.

\paragraph{Approximate duality and the log-rank conjecture.}
Let us now relate the approximate duality conjecture with the log-rank conjecture. By Theorem~\ref{thm:NW_disc}, if $\rank(M_f)=r$ (where the rank is over the reals) we may assume (by potentially restricting $f$ to
a large rectangle) that $\bias(f) \ge \eps = 1/O(r^{3/2})$. Moreover, $\rank_{\F_2}(f) \le \rank(f) = r$. Equivalently put, there are vectors $a_x,b_y \in \F_2^r$ such that
$$
\ip{a_x,b_y} = f(x,y).
$$
Let us define $A=\{a_x: x \in X\}, B=\{b_y: y \in Y\}$. Then by definition, since $\bias(f) \ge \eps$, the sets $A,B$ are $\eps$-approximate dual. Then, by the approximate duality conjecture, there are large subsets $A' \subset A, B' \subset B$ such that $\ip{a,b}$ is constant for all $a \in A', b \in B'$. That is, the rectangle $A' \times B'$ is monochromatic! Working out the parameters, the approximate duality conjecture implies that
$M_f$ contains a monochromatic rectangle $R$ of size $|R| \ge \exp(-O(\sqrt{r \log(r)})) |X \times Y|$. As this holds for any matrix of rank $r$, Theorem~\ref{thm:NW} implies that $f$ has a deterministic protocol of complexity at most $O(\sqrt{r \log(r)})$. Thus, we obtain the following corollary.

\begin{cor}
If Conjecture~\ref{conj:approx_dual} is true, then any boolean function $f$ with $\rank(f)=r$ has a deterministic protocol of complexity $O(\sqrt{r \log(r)})$.
\end{cor}

Of course, we do not know if Conjecture~\ref{conj:approx_dual} is true or not. Ben-Sasson and Ron-Zewi proved the following weak version of it, which has no direct implication for the log-rank conjecture.

\begin{thm}[\cite{zewi2011affine}]\label{thm:BZ}
For any $\alpha>0$ there exist $\eps>0$ such that the following holds. Let $A,B \subset \F_2^r$ be sets which are $(1-\eps)$-approximate dual. Then there exist subsets $A' \subset A, B' \subset B$ and a value $c \in \F_2$ such that
$$
\ip{a,b}=c \qquad \forall a \in A', b \in B',
$$
where
$$
\frac{|A|}{|A'|}, \frac{|B|}{|B'|} \le 2^{\alpha r}.
$$
\end{thm}

Ben-Sasson, Ron-Zewi and the author~\cite{BLR12_An_Ad} proved a slightly stronger version, assuming a number-theoretic conjecture known as the polynomial Freiman-Ruzsa conjecture. This conjecture can be defined over arbitrary Abelian groups, but we only need it for the additive group $\F_2^n$.
\begin{conj}[The polynomial Freiman-Ruzsa conjecture over $\F_2^n$]\label{conj:PFR}
Let $A \subset \F_2^n$ be a set, and let $A+A=\{a_1+a_2: a_1,a_2 \in A\}$ be its sumset. If $|A+A| \le K|A|$ then there exists an affine subspace $V \subset \F_2^n$
of size $|V| \le |A|$ such that
$$
|A \cap V| \ge K^{-O(1)} |A|.
$$
\end{conj}
The polynomial Freiman-Ruzsa conjecture is one of the fundamental open problems in additive combinatorics, see e.g.~\cite{green2004finite} for a discussion of the conjecture. A quasi-polynomial analog of it was proved by Sanders~\cite{sanders2010bogolyubov}, see also~\cite{lovett2012exposition} for an exposition. If one assumes Conjecture~\ref{conj:PFR} to hold, Ben-Sasson et al~\cite{BLR12_An_Ad} proved an improved bound on the approximate duality conjecture.

\begin{thm}[\cite{BLR12_An_Ad}]\label{thm:BLR_approxdual}
Assume that the polynomial Freiman-Ruzsa conjecture over $\F_2^n$ (Conjecture~\ref{conj:PFR}) is true. Let $A,B \subset \F_2^r$ be sets which are $\eps$-approximate dual for $\eps \ge 2^{-\sqrt{r}}$. Then there exist subsets $A' \subset A, B' \subset B$ and a value $c \in \F_2$ such that
$$
\ip{a,b}=c \qquad \forall a \in A', b \in B',
$$
where
$$
\frac{|A|}{|A'|}, \frac{|B|}{|B'|} \le 2^{O(r/\log(r))}.
$$
\end{thm}

Theorem~\ref{thm:BLR_cc} follows as an immediate corollary from the combination of Theorem~\ref{thm:BLR_approxdual} with Theorem~\ref{thm:NW}. We restate it below for the convenience of the reader.

\restate{Theorem~\ref{thm:BLR_cc}}{
Assuming the polynomial Freiman-Ruzsa conjecture over $\F_2^n$, for any boolean function $f$,
$$
\CCdet(f)\le O(\rank(f) / \log \rank(f)).
$$

}

\paragraph{Approximate duality and two-source extractors.}
The original application of~\cite{zewi2011affine} for the approximate duality conjecture was for the construction of pseudo-random graphs, specifically construction of two-source extractors from certain constructions of two-source dispersers. In the following, we focus for simplicity on the case of dispersers and extractors which output a single bit, and we somewhat abuse the standard notations in this field. Let $G=(U,V,E)$ be a bi-partite graph. The graph $G$ is a $k$-Ramsey graph (also called a disperser), if it contains no bi-partite clique or independent set of size $k \times k$. Equivalently, for any subsets $A \subset U, B \subset V$ of size $|A|=|B|=k$, if we denote by $E(A,B)$ the set of induced edges between $A$ and $B$, then
$$
1 \le |E(A,B)| \le |A||B|-1.
$$
The graph is called a $(k,\eps)$ two-source extractor if in fact the number of edges between $A,B$ is close to what might be expected in a random graph, that is
$$
(1/2-\eps)|A||B| \le |E(A,B)| \le (1/2+\eps)|A||B|.
$$
Ben-Sasson and Ron-Zewi~\cite{zewi2011affine} showed that certain constructions of Ramsey graphs are inherently also two-source extractors for weaker parameters.
Consider the following construction of a bi-partite graph $G=(U,V,E)$: $U,V \subset \F_2^n$, and for $u \in U, v \in V$ we have $(u,v) \in E$ if $\ip{u,v}=1$. Assume that $G$ is not a $(k,\eps)$ two-source extractor. That is, there are subsets $A \subset U, B \subset V$ of size $|A|=|B|=k$ such that (say) $|E(A,B)| \ge (1/2+\eps) |A||B|$. This means that the approximate duality measure between $A,B$ is at least $2\eps$, which by the approximate duality conjecture (Conjecture~\ref{conj:approx_dual}) implies that we can find large subsets $A' \subset A, B' \subset B$ such that (say) $|E(A',B')|=0$. Then, we conclude that the graph $G$ is not a $k'$-Ramsey graph for $k'=\min(|A'|,|B'|)$. Otherwise put, any bi-partite graph, constructed in this way, which is $k'$-Ramsey, must also be a $(k,\eps)$ two-source extractor, where $k$ is somewhat larger than $k'$. For further details we refer the reader to the original paper~\cite{zewi2011affine}.

\section{From randomized to deterministic protocols}
\label{sec:rand_det}

The log-rank conjecture speculates that low rank boolean functions have efficient deterministic protocols. We already saw in Theorem~\ref{thm:NW} that a sufficient condition is that any low rank boolean matrix contains a large monochromatic rectangle. Here, we describe another reduction, due to Gavinsky and the author~\cite{GL13:lowrank_equiv}. We will show that it is also sufficient to construct a randomized protocol computing the function.

A randomized protocol computing a function $f(x,y)$ is a protocol, in which both parties are allowed to use randomized strategies, such that for every input $x,y$, the protocol computes the correct value $f(x,y)$ with probability at least $2/3$. Note that a randomized protocol is a distribution over deterministic protocols. The complexity of a randomized protocol is the maximal number of bits that may be sent by the protocol. We recall Theorem~\ref{thm:GL} for the convenience of the reader.

\restate{Theorem~\ref{thm:GL}}{
If a boolean function has a randomized protocol of complexity $c$, then it also has a deterministic protocol of complexity $O(c \cdot \log^2(\rank(f)))$.
}

\begin{proof}
Let $p(x,y)$ denote the probability that the protocol computes $f$ correctly on inputs $x,y$, where by assumption $p(x,y) \ge 2/3$. We can increase the success probability by repeating the protocol a few times, and computing the majority of the values obtained. Specifically, if we repeat the protocol $O(\log 1/\eps)$ times, we obtain a randomized protocol which uses $c' = O(c \log(1/\eps))$ bits and computes $f(x,y)$ correctly with probability $1-\eps$. A randomized protocol is a distribution over deterministic protocols; hence, if we consider the uniform distribution over inputs, we get by an averaging argument that there exists a deterministic protocol $\pi(x,y)$ of complexity $c'$ such that
$$
\big|\{(x,y) \in X \times Y: \pi(x,y)=f(x,y)\}\big| \ge \left(1-\eps\right)|X \times Y|.
$$
A deterministic protocol of complexity $c'$ corresponds to a partition to $N=2^{c'}$ many rectangles. We next argue that there exists a large rectangle on which $f$ is nearly fixed. Let $R_1,\ldots,R_N$ denote the rectangles corresponding to the protocol $\pi$. Denote by $\mu(R)=|R|/|X \times Y|$ the fractional size of a rectangle, and by $\alpha(R)=|\{(x,y) \in R: \pi(x,y) \ne f(x,y)\}|/|R|$ the fraction of elements in $R$ on which the protocol $\pi$ makes a mistake. By assumption, we have
$$
\sum_{i=1}^N \mu(R_i)=1;\qquad \sum_{i=1}^N \mu(R_i) \alpha(R_i) \le \eps.
$$
One can verify that these imply that there must be a rectangle $R=R_i$ such that
$$
\mu(R) \ge 1/2N;\qquad \alpha(R) \le 2\eps.
$$
As $\pi$ is fixed on $R$, we can assume without loss of generality that
$$
|\{(x,y) \in R: f(x,y) = 1\}| \ge (1-2 \eps)|R|.
$$
Let $r=\rank(f)$. We next show that by setting $\eps = 1/8r$, there exists a large sub-rectangle $R' \subset R$ on which $f$ is monochromatic.

\begin{claim}\label{claim:nearmono_to_mono}
Let $f$ be a boolean function of rank $r$, and assume there exists a rectangle $R$ on which $f(x,y)=1$ for at least $1-1/4r$ of the elements in $R$. Then, there exists a sub-rectangle $R' \subset R$
of size $|R'| \ge |R|/8$ such that $f(x,y)=1$ for all $(x,y) \in R'$.
\end{claim}

\begin{proof}
Let $R=A \times B$.
Let $A' \subset A$ be the set of rows for which at most $1/2r$ fraction of the elements are $-1$,
$$
A'=\big\{x \in A: \left|\{y \in B: f(x,y)=-1\}\right| \le |B|/2r\big\}.
$$
By Markov inequality, $|A'| \ge |A|/2$. Let $x_1,\ldots,x_r \in A'$ be indices so that their rows span $f$ restricted to $A' \times B$. Let
$$
B'=\{y \in B: f(x_1,y)=\ldots=f(x_r,y)=1\}.
$$
Since each of the rows $x_1,\ldots,x_r$ contain at most $1/2r$ fraction of elements which are $-1$ we have $|B'| \ge |B|/2$. Now, this implies that all
rows in $A' \times B'$ are either the all $1$ or all $-1$. Choosing the largest half gives the required rectangle. This gives a monochromatic rectangle $R' \subset R$ of size
$|R'| \ge |R|/8$.
\end{proof}

To conclude, we would like to apply Theorem~\ref{thm:NW} in order to show the existence of a deterministic protocol. The reader can verify, that although the conditions of Theorem~\ref{thm:NW} require one to show that any low rank function has a large monochromatic rectangle, in fact for the proof to go through, it suffices to assume that this holds only for functions which are restrictions of $f$ to rectangles.
The same argument as above shows that for any rectangle $R \subset X \times Y$, there exists a sub-rectangle
$R' \subset R$ of size $|R'| \ge 2^{-O(c \log(r))} |R|$ on which $f$ is monochromatic. Note that, as the bound $c$ does not improve as the rank decreases, we incur an additional multiplicative factor of $\log(r)$ in the communication complexity. We deduce that there exists a deterministic protocol computing $f$ of complexity $O(c \log^2(r))$, as claimed.
\end{proof}

\section{Discrepancy of matrices and the log-rank conjecture}
\label{sec:disc}

Let $f:X \times Y \to \{-1,1\}$ be a boolean function. For a distribution $\mu$ on $X \times Y$, the \emph{discrepancy} of $f$ with respect to $\mu$ is the maximal correlation that $f$ has with rectangles,
$$
\disc(f;\mu) = \max_R \left| \sum_{(x,y) \in R} f(x,y) \mu(x,y) \right|
$$
where $R$ ranges over all rectangles. The discrepancy of $f$ is its discrepancy for the worse case distribution,
$$
\disc(f) = \min_{\mu} \disc(f;\mu).
$$
Discrepancy is a well-studied property in the context of communication complexity lower bounds, see e.g. the survey~\cite{Lokam09:book} for details. On the other hand, it is known that low-rank boolean matrices have noticeable discrepancy~\cite{LMSS07_Com,LS09_Lea}: if $f$ has rank $r$ then
\begin{equation}\label{eq:disc_lowrank}
\disc(f) \ge \frac{1}{8\sqrt{r}}.
\end{equation}
A result of the author~\cite{lovett2013communication} shows that discrepancy can be used to prove upper bounds as well. We restate Theorem~\ref{thm:lovett} for the convenience of the reader.

\restate{Theorem~\ref{thm:lovett}}{
For any boolean function $f$,
$$
\CCdet(f)\le O\left(\sqrt{\rank(f)} \cdot \log \rank(f)\right).
$$
}

The following lemma is the main technical tool. It shows that a function with high discrepancy contains a large rectangle which is almost monochromatic. In fact, this is true with respect
to any distribution over the inputs. We make the following definitions: given a distribution $\mu$ over $X \times Y$, let $\mu(R) = \sum_{(x,y) \in R} \mu(x,y)$ denote the probability of an input landing in $R$, and
$\E_{\mu}[f] = \sum_{(x,y) \in X \times Y} \mu(x,y) f(x,y)$ the average of $f$ with respect to $\mu$. For a rectangle $R$ such that $\mu(R)>0$, let $\mu|R$ the distribution $\mu$ conditioned on being in $R$, that is,
$(\mu|R)(x,y)=1_{(x,y) \in R} \cdot \mu(x,y) / \mu(R)$.

\begin{lemma}\label{lemma:disc_nearmono}
Let $f:X \times Y \to \{-1,1\}$ be a function with $\disc(f)=\delta$. Then for any $\eps>0$ and any distribution $\mu$ over $X \times Y$, there exists a rectangle $R$ with
$$
\mu(R) \ge 2^{-O(\delta^{-1} \cdot \log(1/\eps))}
$$
such that $\big| \E_{\mu|R}[f] \big| \ge 1-\eps$.
\end{lemma}

\begin{proof}[Proof of Theorem~\ref{thm:lovett}, assuming Lemma~\ref{lemma:disc_nearmono}]
Let $f$ be any boolean function of rank $r$. Apply Lemma~\ref{lemma:disc_nearmono} with $\mu$ the uniform distribution over $X \times Y$, $\delta \ge 1/8\sqrt{r}$ and $\eps=1/4r$, to deduce the existence of a rectangle $R \subset X \times Y$ of size $|R| \ge 2^{-O(\sqrt{r} \log(r))} |X \times Y|$ such that $f(x,y)=v$ for $1-1/4r$ fraction of elements in $R$. Apply Claim~\ref{claim:nearmono_to_mono} to deduce that there exists a sub-rectangle $R' \subset R$ of size $|R'| \ge |R|/8$ on which $f$ is monochromatic. By Theorem~\ref{thm:NW}, this implies that any function of rank $r$ has a deterministic protocol of complexity $O(\sqrt{r} \log(r))$.
\end{proof}

We now turn to prove Lemma~\ref{lemma:disc_nearmono}. The proof of Lemma~\ref{lemma:disc_nearmono} which we give below is a simplification of the original proof of~\cite{lovett2013communication}, which was presented to us by Salil Vadhan~\cite{Vadhan:personal}.

\begin{proof}[Proof of Lemma~\ref{lemma:disc_nearmono}]

Let us assume without loss of generality that $\E_{\mu}[f] \ge 0$, otherwise apply the lemma to $-f$.
Let $\sigma$ be any distribution over $X \times Y$ such that $\E_{\sigma}[f]=0$. By assumption, there exists a rectangle $R_1$ such that
$$
\left|\sum_{(x,y) \in R_1} \sigma(x,y) f(x,y)\right| \ge \delta.
$$
Let $R_1 = A \times B$ and define $A' = X \setminus A, B' = Y \setminus B$. Consider the four rectangles
$$
R_1=A \times B, R_2 = A' \times B, R_3 = A \times B', R_4 = A' \times B'.
$$
As $\sum_{(x,y) \in X \times Y} \sigma(x,y) f(x,y) = \E_{\sigma}[f]=0$, there must exist a rectangle $R \in \{R_1,R_2,R_3,R_4\}$ such that
$$
\sum_{(x,y) \in R} \sigma(x,y) f(x,y) \ge \delta/3.
$$
This holds for any distribution $\sigma$ for which $\E_{\sigma}[f]=0$. Hence, we can apply von Neumann's Minimax Theorem~\cite{neumann1928theorie} and deduce that there exists a distribution $\rho$ over rectangles, such that for any distribution $\sigma$ for which $\E_{\sigma}[f]=0$, we have
$$
\E_{R \sim \rho}\left[ \sum_{(x,y) \in R} \sigma(x,y) f(x,y) \right] \ge \delta/3.
$$
Equivalently,
$$
\sum_{(x,y) \in X \times Y} \Pr_{R \sim \rho}[(x,y) \in R] \cdot \sigma(x,y) f(x,y) \ge \delta/3.
$$

Fix $(x_1,y_1) \in f^{-1}(1)$ and $(x_2,y_2) \in f^{-1}(-1)$. Let $\sigma$ be the distribution given by $\sigma(x_1,y_1)=\sigma(x_2,y_2)=1/2$. As $\E_{\sigma}[f]=0$ we have
$$
\Pr_{R \sim \rho}[(x_1,y_1) \in R] - \Pr_{R \sim \rho}[(x_2,y_2) \in R] \ge (2/3) \delta.
$$
Let $p$ be the \emph{minimal} probability that $(x_1,y_1) \in R$ over all $(x_1,y_1) \in f^{-1}(1)$, where $R$ is sampled according to $\rho$; and let $q$ be the \emph{maximal} probability that $(x_2,y_2) \in R$ over all $(x_2,y_2) \in f^{-1}(-1)$. We established that
$$
p - q \ge (2/3) \delta.
$$

Fix $t \ge 1$ and let $R_1,\ldots,R_t \sim \rho$ be chosen independently, and let $R^* = R_1 \cap \ldots \cap R_t$ be their intersection. We will show that for an appropriate choice of $t$, the rectangle $R^*$ satisfies the requirements of the lemma with positive probability (and hence such a rectangle exists). We will use the fact that for any $x \in X, y \in Y$,
$$
\Pr[(x,y) \in R^*] = \Pr_{R \sim \rho}[(x,y) \in R]^t.
$$
Consider the random variable
$$
T = \mu(R^*) - (1/\eps) \cdot \mu(R^* \cap f^{-1}(-1)).
$$
By linearity of expectation, we have
\begin{align*}
\E[T] &= \sum_{(x,y) \in f^{-1}(1)} \mu(x,y) \Pr[(x,y) \in R^*] - \sum_{(x,y) \in f^{-1}(-1)} \mu(x,y) ((1/\eps)-1) \Pr[(x,y) \in R^*] \\
& \ge \mu(f^{-1}(1)) \cdot p^t - \mu(f^{-1}(-1)) \cdot q^t / \eps \\
& \ge 1/2 \cdot (p^t - q^t / \eps),
\end{align*}
where we used our initial assumption that $\E_{\mu}[f]=\mu(f^{-1}(1))-\mu(f^{-1}(-1)) \ge 0$. Setting $t=O(p/\delta \cdot \log(1/\eps))$ gives
$$
q^t / p^t \le (1 - (2/3) \delta/p)^t \le \eps/2.
$$
For this choice of $t$, we have
$$
\E[T] \ge p^t / 4 = 2^{-O(\delta^{-1} \cdot \log(1/\eps))}.
$$
Let $R^*$ be a rectangle which achieves this average, that is
$$
\mu(R^*) - (1/\eps) \cdot \mu (R^* \cap f^{-1}(-1)) \ge 2^{-O(\delta^{-1} \cdot \log(1/\eps))}.
$$
In particular, we learn that both $\mu(R^*) \ge 2^{-O(\delta^{-1} \cdot \log(1/\eps))}$ (which satisfies the first requirement) and furthermore that $\mu(R^* \cap f^{-1}(-1)) \le \eps \cdot \mu(R^*)$, which implies that $\E_{\mu|R^*}[f] \ge 1-\eps$ (which satisfies the second requirement).
\end{proof}

\section{Further research}
\label{sec:further}

There are several directions for further research. We describe a few concrete ones below.

\subsection{Randomized protocols vs approximate rank}
The \emph{approximate rank} of a boolean function $f(x,y)$ is the minimal rank of an $X \times Y$ real matrix $M$ such that
$$
2/3 \le M(x,y) f(x,y) \le 1.
$$
Similar to the log rank lower bound for the deterministic communication complexity, the log of the approximate rank is a lower bound on the randomized communication complexity of a function. The log-rank conjecture for randomized protocols speculates that it is also an upper bound, up to polynomial factors. As a first step, one can attempt to generalize Theorem~\ref{thm:lovett} to approximate rank and randomized protocols.

\begin{prob}
Let $f$ be a boolean function with approximate rank $r$. Show that $f$ has a randomized protocol of complexity $\sqrt{r} \cdot \poly \log(r)$.
\end{prob}

\subsection{Quantum protocols for low-rank matrices}
The work of~\cite{GL13:lowrank_equiv} shows that if low-rank functions have certain types of efficient protocols (randomized protocols, low information cost protocols, or zero-communication protocols), then up to a poly-logarithmic factor in the rank, they also have efficient deterministic protocols. One type of protocol which they were not able to analyze is quantum protocols. This is interesting on its own right, but also because to the best of our current knowledge, it may be that quantum protocols are only polynomially better than randomized protocols, for any complete boolean function (exponential separations are known for partial functions, see e.g.~\cite{raz1999exponential,regev2011quantum}). Thus, understanding quantum protocols, even just for low-rank functions, seems like an important step towards a better understanding of quantum protocols in general.

\begin{prob}
Let $f$ be a boolean function which can be computed by a quantum protocol of complexity $c$. Show that $f$ can also be computed by a deterministic protocol of complexity $c \cdot \poly \log (\rank(f))$.
\end{prob}

\subsection{The structure of low-rank sparse matrices, and matrix rigidity}

The proof of Theorem~\ref{thm:lovett} applies to boolean matrices. We conjecture in~\cite{lovett2013communication} that it can be generalized to show that any low rank sparse matrix contains a large zero rectangle.

\begin{conj}\label{conj:sparse}
Let $M$ be an $n \times n$ real matrix with $\rank(M)=r$ and such that $M_{i,j} \ne 0$ for at most $\eps n^2$ entries. Then, there exist $A,B \subset [n]$ such that
$$
M_{a,b}=0 \qquad \forall a \in A, b \in B
$$
such that $|A|,|B| \ge n \cdot \exp(-O(\sqrt{\eps r}))$.
\end{conj}

The reader can observe the similarities of Conjecture~\ref{conj:sparse} to the approximate duality conjecture (Conjecture~\ref{conj:approx_dual}) which we discussed. Note that here we
consider the case where nearly all the elements are zero, while in the approximate duality conjecture we only assumed a small bias. Nevertheless, the same construction shows that the bounds in Conjecture~\ref{conj:sparse}, if true, are the best possible.

A matrix $M$ is called $(r,s)$-rigid, if its rank cannot be made smaller than $r$ by changing at most $s$ entries in $M$. The problem of explicitly constructing rigid matrices was
introduced by Valiant~\cite{Valiant:rigidity} in the context of arithmetic circuits lower bounds, and was also studied by Razborov~\cite{Razborov89:rigid_cc} in
the context of separation of the analogs of PH and PSPACE in communication complexity. Despite much research, the best results to date are achieved by the so-called
"untouched minor" argument, which gives explicit matrices which are $(r,s)$-rigid with $s = \Omega\left(\frac{n^2}{r} \log\left(\frac{n}{r}\right)\right)$. See e.g. the
survey of Lokam~\cite{Lokam09:book} for details. We will prove the following corollary of Conjecture~\ref{conj:sparse}, which improves previous bounds by a logarithmic factor.

\begin{cor}
Assuming Conjecture~\ref{conj:sparse}, there exists an explicit $n \times n$ real matrix which is $(r,s)$-rigid for $s=\Omega\left(\frac{n^2}{r} \log^2\left(\frac{n}{r}\right)\right)$.
\end{cor}

\begin{proof}
Let $M$ be an $n \times n$ matrix of rank $r$, such that all $r \times r$ minors of $M$ have full rank. For example, such a matrix may be constructed as $M= N N^t$ where $N$ is an $n \times r$
matrix such that any $r$ rows of $N$ are linearly independent. Assume that $M$ is not $(r,s)$-rigid. Then, we can decompose
$$
M=L+S,\quad \rank(L)<r, \quad S \textrm{ is s-sparse}.
$$
Let $s=\eps n^2$. The matrix $S$ is both $s$-sparse and low rank, as $\rank(S) \le \rank(M)+\rank(L) < 2r$. Hence, by Conjecture~\ref{conj:sparse}, there exist $A,B \subset [n]$ of size
$|A|,|B| \ge n \cdot \exp(-O(\sqrt{\eps r}))$ such that $S_{a,b}=0$ for all $a \in A, b \in B$. Hence, $M_{a,b}=L_{a,b}$. If $|A|,|B| \ge r$, we must have that $\rank(L) \ge \rank(M)=r$. So,
$n \cdot \exp(-O(\sqrt{\eps r})) < r$ and the corollary follows by rearranging the terms.
\end{proof}

\paragraph{Acknowledgements}
I thank V. Arvind for inviting me to write this survey. I thank Dmitry Gavinsky, Noga Ron-Zewi and Adi Shraibman for helpful comments on earlier versions of this manuscript.

\bibliographystyle{alpha}
\bibliography{log-rank}

\end{document}